\documentclass[3pt,time,rewiew]{elsarticle}

\usepackage{amssymb}
\usepackage{titlesec} 
\usepackage{amsmath}
\usepackage{amsthm}
\pdfoutput=1
\usepackage{multirow}

\usepackage{comment}
\usepackage{subfigure}
\usepackage{colortbl}

\usepackage{tkz-graph}
\usepackage{pgf, tikz}
\usetikzlibrary{graphs}
\usetikzlibrary{trees}

\usepackage[algoruled,longend,commentsnumbered,linesnumbered,ruled,vlined]{algorithm2e/algorithm2e}

\newcolumntype{P}[1]{>{\centering\arraybackslash}p{#1}} 

\newtheorem{teo}{Theorem}[section]
\newtheorem{lemma}{Lemma}[section]

\journal{Information Processing Letters}

\begin{document}

\begin{frontmatter}

\title{Choosability in bounded sequential list coloring}

\author[rva,rvt]{Simone Gama\corref{cor1}}
\ead{simone.gama@icomp.ufam.edu.br}

\author[rva,rvt]{Rosiane de Freitas}
\ead{rosiane@icomp.ufam.edu.br}

\author[rva,rvt]{M\'ario Salvatierra}
\ead{mario@icomp.ufam.edu.br}


\address[rva]{UFAM - Federal University of Amazonas, Manaus-AM, Brazil}


\address[rvt]{Institute of Computing, ICOMP-UFAM}

\begin{abstract}
The list coloring problem is a variation of the classical vertex coloring problem, extensively studied in recent years, where each vertex has a restricted list of allowed colors, and having some variations as the $(\gamma,\mu)$-coloring, where the color lists have sequential values with known lower and upper bounds. This work discusses the choosability property, that consists in determining the least number $k$ for which it has a proper list coloring no matter how one assigns a list of $k$ colors to each vertex. This is a $\Pi_2^P$-complete problem, however, we show that $k$-$(\gamma,\mu)$-choosability is an $NP$-problem due to its relation with the $k$-coloring of a graph and application of methods of proof in choosability for some classes of graphs,  such as complete bipartite graph, which is $ 3 $-choosable, but $ 2 $-$(\gamma,\mu)$-choosable.
\end{abstract}


\end{frontmatter}


\section{Introduction}\label{introduction}

Let $G=(V,E)$ be a simple graph, where $V$ is the set of vertices and $E$ is the set of edges. A \textit{$k$-coloring} of $G$ is an assignment of colors $\{1, 2, \dots, k\}$ to the vertices of $G$ such that no two adjacent vertices share the same color. A bipartite graph is a graph whose vertices can be divided into two disjoint and independent sets $ W $ and $ V $  such that every edge connects a vertex in $ W $ to one in $ V $. An outerplanar graph is a graph that has a planar drawing for which all vertices belong to the outer face of the drawing \cite{book-chartrand}. The chromatic number $\chi_G$ of a graph is the minimum value of $k$ for which $G$ is $k$-colorable. The classic graph coloring problem, which consists in finding the chromatic number of a graph, is one of the most important combinatorial optimization problems and it is known to be $NP$-complete. There are several versions of this classic vertex coloring problem, involving additional constraints, in both edges as vertices of the graph. One of them is the list coloring problem, where given a graph $G$ there is an associated set $L(v)$ of allowed color lists for each vertex of $v \in V(G)$. If it is possible to get a proper coloring of $G$ with these color lists, then we have a list coloring of $G$. It was introduced by Paul Erd\"{o}s \textit{et al} in 1979~\cite{article-erdos}, and independently by Vizing in 1976~\cite{article-vizing}. The list coloring problem also has its variations, among them the $\mu$-coloring and the $(\gamma,\mu)$-coloring, introduced by Bonomo~\textit{et al.} \cite{article-bonomo-1,bonomo2009exploring,gama2016problemas}. Give a graph $G$ and a function $\gamma,\mu:V(G)\rightarrow \mathbb{N}$ such that $\gamma(v)\leq\mu(v)$ for every $v\in V(G)$, $G$ is $(\gamma,\mu)$-colorable if there exists a coloring $f$ de $G$  such that $\gamma(v)\leq f(v)\leq\mu(v)$ for every $v\in V(G)$. In Table~\ref{table_3:complexidade} we present brief comparative results from the literature in computational complexity of coloring problems for some graph classes.

\begin{table}[h]
	\centering
	\footnotesize{		
		\begin{tabular}{|p{2.2cm}|p{1.6cm}|p{1.6cm}|p{2cm}|p{1.6cm}|}
			\hline
			\centering\textbf{Graph classes}  & \multirow{1}{*}{$k$-coloring}  	& \centering$\mu$-coloring & \centering$(\gamma,\mu)$-coloring & {\centering List coloring} \\ 
			\hline
			General      & NP-C \cite{book-garey}& NP-C \cite{article-bonomo-1} & NP-C \cite{bonomo2009exploring}&  NP-C \cite{dror1999complexity}\\ \hline
			
			Bipartite  & P \cite{bonomo2009exploring}& NP-C \cite{article-bonomo-1}& NP-C \cite{bonomo2009exploring}& NP-C \cite{kubale2004graph}\\ \hline
			
			Split      			& P \cite{golumbic2004algorithmic}& NP-C \cite{bonomo2009exploring} & NP-C \cite{bonomo2009exploring}& NP-C \cite{bonomo2009exploring}\\ \hline
			
			Interval   & P \cite{grotschel1981ellipsoid} & NP-C \cite{bonomo2009exploring}& NP-C \cite{bonomo2009exploring}& NP-C \cite{bonomo2009exploring}\\ \hline
			
			Line of $K_n$& P \cite{konig1916graphen}& NP-C \cite{bonomo2009exploring}& NP-C \cite{bonomo2009exploring}& NP-C \cite{kubale2004graph} \\ \hline
			
			Line of $K_{n,n}$& P \cite{konig1916graphen}& NP-C \cite{bonomo2009exploring}& NP-C \cite{bonomo2009exploring}& NP-C \cite{bonomo2009exploring} \\ \hline
			
			Cographs   & P \cite{grotschel1981ellipsoid} &  P \cite{article-bonomo-1}&   \textit{open}  &  NP-C \cite{jansen1997generalized}\\ \hline
			
			$K_{n,n}$  & P \cite{bonomo2009exploring} &  P \cite{article-bonomo-1}&   P \cite{bonomo2009exploring}  &  NP-C \cite{bonomo2009exploring}\\ \hline
			
			Blocks    & P \cite{bonomo2009exploring} & P \cite{bonomo2009exploring} & P \cite{bonomo2009exploring}& P \cite{bonomo2009exploring}\\ \hline
		\end{tabular}
	\label{table_3:complexidade}
	\caption{Computational complexity between vertex coloring problems.} 
	}
\end{table}

We study algorithms in choosability in graphs (in Section \ref{sec:choos}) and we studied the correlation between the choosability in graphs and $ (\gamma, \mu) $-coloring, thus being $ k $-$ (\gamma, \mu) $-choosable. We show that $ k $-$ (\gamma, \mu) $-choosable for bipartite graphs and outerplanar graphs (in Section \ref{sec:results}) and further, we show that a graph with a $ k $-coloring has a $ k $-$ (\gamma, \mu) $-choosable, which strengthens the results.

\section{Algorithms for choosability in graphs} \label{sec:choos}
Given a graph $G$ and a set $L(v)$ of colors for each vertex $v$, a list of allowed colors, a list coloring of $G$ is a choice function that maps every vertex $v$ to a color in the list $L(v)$. As with vertex coloring problem, a list coloring is generally assumed to be proper, meaning no two adjacent vertices receive the same color. Thus, a graph is $k$-choosable (or $k$-list-colorable) if it has a proper list coloring no matter how one assigns a list of $k$ colors to each vertex. The choice number (or list chromatic number), denoted by $\chi_{\ell}(G)$, is the least number $k$ such that $G$ is $k$-choosable. A more general case is the $f$-choosability. Given a function $f$ of assigning a positive integer $f(v)$ to each vertex $v$, a graph $G=(V,E)$ is $f$-choosable (or $f$-list-colorable) if it has a list coloring no matter how one assigns a list of $f(v)$ colors to each vertex $v$. A particular case happens when $f(v) = k$, for all vertices $v \in V(G)$, when $f$-choosability corresponds to $k$-choosability.

\subsection{Computational complexity for the $k$-choosability}
The $k$-choosability of a graph $G$ is one of the few problems known to be $\Pi_2^P$-complete and this gives the fact that there is an exponential number of possible color lists with arbitrary and not necessarily consecutive values, of certain size (details of $\Pi_2^P$-complete in \cite{book-garey}). Thus, to determine the $k$-choosability, we have to generate all possible instances with lists of size $k$, that is, to check for each possible distribution of color lists of size $k$ to the vertices, an arrangement with repetition. If for a given instance of size $k$ of color lists, there is no solution, then, we repeat the process for color lists of size $k+1$ until the property is satisfied. The aim is to determine the smallest $k$ for which we have a list coloring for all possible instances. An algorithm for $k$-choosable enumerates all color lists of size $ k $ to check if the graph is colorable for all these lists. The solutions are formed by the color lists, taken $ n $ by $ n $, with repeated elements (see Algorithm \ref{pseudo_alg:e_k-choosable}).

\begin{algorithm}[H]\label{pseudo_alg:e_k-choosable}
	\DontPrintSemicolon
	$X\leftarrow \binom{n}{k}$ \label{algo:combinacao} \\
	$P_x \leftarrow$ all $n$-permutations of $X$ with repeated elements;\label{algo:permutacao} \\
	
	\For{$L \in P_x$}
	{
		\If {!(Exists\_List\_Coloring($G, L, not, 1$))}
		{
			\textbf{Return} NOT
		}
	}
	\textbf{Return} YES
	\caption{Is\_$k$-choosable($G=(V,E), k$)}
\end{algorithm}
\vspace{1em}

The previous algorithm uses another algorithm for list coloring problem because it needs to check in each solution whether they have a valid color list (see Algorithm \ref{Pseudo_alg:lista_coloracao_exato}). The algorithm is exponential time, that is, it checks in a graph $ G $ with lists of colors assigned to its vertices all the possibilities of having a coloration. If not, the algorithm will attempt this check with the next set of color lists until there are no more sets to be tested. Thus, this process is a \textit{for all} lists "there exists" a feasible solution, which configures the case of a $\Pi_2^P$-complete problem.

\begin{algorithm}[H]\label{Pseudo_alg:lista_coloracao_exato}
	\DontPrintSemicolon
	\footnotesize
	\If{not exists}
	{\label{alg_exato:existe}
		\uIf {$L= \varnothing$}
		{\label{alg_exato:lista_vazia}
			\If {ExistsColoring($G$)}
			{\label{alg_exato:existe_solucao}
				exists$\leftarrow$YES
			}
		}\Else{
			$l \longleftarrow$ remove the first on the list $L$ \label{alg_exato:remove_lista} \\
			\For{color $\in l$}
			{
				$G.v_i.cor \longleftarrow$ color\label{alg_exato:cor_vertice}\\
				List\_Coloring($G, L, exists, i+1$)\\
				\If{exists}
				{
					\textbf{Break}
				}
			}
		}
	}        
	\caption{Exists\_List\_Coloring($G, L, \&exists, i$)}
\end{algorithm}

\subsection{Additional results applying methods of proof in choosability in graphs}\label{sec:methods}

The work of Woodall \cite{woodall2001list} presents a survey with the main methods of choosability in graphs. These proof methods can be divided into methods that involve digraphs (Alon Tarsi method and kernel's method) and do not involve digraphs (degeneracy in graphs, the method of Hall's theorem and boundary method). The method of Alon and Tarsi was introduced by Alon and Tarsi \cite{alon1992colorings}. The authors using the polynomial graph technique in their proofs. This method is used by Alon and Tarsi to prove that every planar bipartite graph is $ 3 $-choosable. The kernel's method was introduced by Galvin \cite{galvin1995list}. 

The method of degeneracy in graphs is used in some cases of choosability. A graph $ G $ is called $ k $-degenerate if every subgraph $ H $ of $ G $ contains a vertex with degree (in $ H $) at most $ k $. This method is used to show that every planar graph is $ 6 $-choosable and every triangle-free planar graph are $ 4 $-choosable \cite{kratochvil1994algorithmic}. The method of Hall's theorem \cite{hall1948distinct} of the theorem on systems of distinct representatives. This method is used to show that the complete $ k $-partite graph $ K_{2,2,\ldots,2} $ is $ k $-choosable \cite{article-erdos}. Let $S$ be a family of finite sets $S_1 , S_2 , \ldots, S_m$, a system of distinct representatives is a set $x_1 , x_2,\ldots, x_m$ such that $x_i \in S_i,\forall i$ where $x_i \neq x_j,\forall i\neq j$. A family of sets $S_1 , S_2 , \ldots, S_m$ system has a different representatives if and only if for any subset  $I \subseteq \{1, 2, \ldots, m\}$ condition is true $$ |\bigcup _{ i\in I}S_i | \geq |I|$$ The boundary method \cite{thomassen1994every}, define the boundary to comprise those uncolored vertices of $ G $ that are adjacent to colored vertices. This method was introduced in Thomassen's famous article~\cite{thomassen1994every} and was used to show that any planar graph is $ 5 $-choosable. For more results of choosability in planar graphs see \cite{bonomo2009exploring,article-choosable1,article-choosable}.

\section{The $k$-$(\gamma,\mu)$-choosability}\label{sec:results}
\vspace{-.2cm}
We correlate $ k $-choosability with $\left(\gamma,\mu \right)$-coloring. A graph $G$ is $k$-$\left(\gamma,\mu \right)$-choosable or $k$-$\left(\gamma,\mu \right)$-list-colorable, if $G$ is $\left(\gamma,\mu \right)$-colorable for each set $L(V)$ of color lists of each $v \in V(G)$, such that $|L(v)| \geq k$ to each vertex $v$, and the lists are the type $\left( \gamma,\mu \right)$. Application of the proof methods presented in the Section \ref{sec:methods}. The application will be in $k$-$(\gamma$,$\mu)$-coloring for bipartite and outerplanar graphs. 

\begin{teo}
	Every outerplanar graph is $3$-$(\gamma$,$\mu)$-choosable.
\end{teo}\label{teo:outerplanar}

\vspace{-.4cm}
\begin{proof}
	Let $ G = K_1 $. This graph has no edges and is trivially outerplanar and $3$-$(\gamma$,$\mu)$-choosable. Assuming that every outerplanar graph with n vertices is $3$-$(\gamma$,$\mu)$-choosable, it will be shown that the coloring list is also valid for an outerplanar graph $ G_{n+1} $ with $ n+1 $ vertices and with a distribution of color lists in its vertices of type $(\gamma$,$\mu)$ of size $ 3 $. There exists a vertex $ v\in G_{n+1}$ with degree up to maximum $ 2 $. Removing this vertex, we have a subgraph $ G_n $ with $ n $ vertices \cite{book-chartrand}. Since $ G_{n+1} $ can be drawn with its vertices in the form of a cycle with chord without crosses, the subgraph can also be drawn in this way. This implies that $ G_n $ has a color list with lists of type $(\gamma$,$\mu)$ of size $ 3 $. Color the vertices of $ G_{n+1} $ with $ 3 $ or more colors. Since the vertices adjacent to vertex $ v $ of degree at maximum $ 2 $ are already colored, consider two cases: Case one: The vertex $ v $ has degree one. In this case, the vertex $ v $ receives a color from the list of colors and the adjacent vertex receives another color. Since the list color has size three, one color will remain. Case two: The vertex $ v $ has degree two. The vertex $ v $ will receive a third color from its list of that is different from its adjacent ones. And so the demonstration is valid.
\end{proof}

For the next theorem, Hall's method will be applied. To show that a complete bipartite graph is $2$-$(\gamma$,$\mu)$-choosable, let's first show two cases: when the subgraph of a bipartite graph has equal lists color among its partitions and when all the lists color are different.

\begin{lemma}
	Consider the complete bipartite graph $K_{n,n}=G[V,W]$ for $n\geq 1$ with partitions  $V=\{v_1,\ldots ,v_n\}$,$W=\{w_1,\ldots ,w_n\}$ and a distribution lists $\left( \gamma,\mu \right)$ of size $ 2 $, $L(v)$ for all $v\in K_{n,n}$ such that $L(v_i)\neq L(v_j)$ if $i\neq j$, is $L(v_i)= L(w_i)$ for all $i=1,\ldots ,n$. Then with these lists, $K_{n,n}$ is $2$-$\left(\gamma,\mu \right)$-list-colorable.
\end{lemma}\label{lemmaOne}

\vspace{-.4cm}
\begin{proof}
	Consider a $K_{n,n}$ graph where the color lists of type $\left(\gamma,\mu \right)$ that occur in one partition also occur in the other partition, see Figure \ref{Figura:GrafoBipartidoLema1}. For $ n = 1 $, the graph $ K_{1,1} $ has two vertices $ v_1 $ and $ w_1 $ and $ L(v_1)=L(w_1)=\{1,2\} $. Suppose now that the theorem holds for $ K_{n,n} $, where $ n\ge1 $ and then holds for $ K_{n+1,n+1} $. Without loss of generality, assume that the lists $ L(vi)=\gamma_i; \gamma_{i+1} $ are distributed at the vertices such that $ \gamma_1 <\gamma_2 <\ldots <\gamma_{n+1} $. 
	
	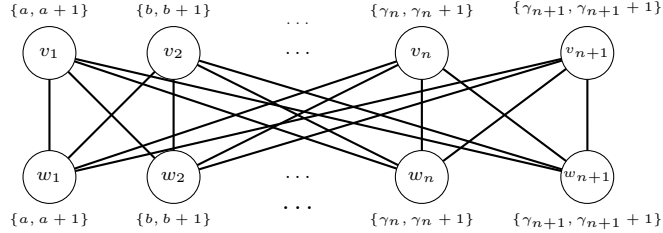
\begin{figure}[ht]
		\centering
		\begin{tikzpicture}[scale=1.1]
		\SetGraphUnit{2}
		\GraphInit[vstyle=Normal]
		\SetUpVertex[LabelOut]
		\tikzset{VertexStyle/.style = {shape = circle, minimum size = 20pt, draw, color=black}}
		
		\Vertex[L=\hbox{\tiny$\{a,a+1\}$},Lpos=90,x=0.0cm,y=1.5cm]{v0}
		\Vertex[L=\hbox{\tiny$\{a,a+1\}$},Lpos=-90,x=0.0cm,y=0.0cm]{v1}
		\Vertex[L=\hbox{\tiny$\{b,b+1\}$},Lpos=-90,x=1.5cm,y=0.0cm]{v2}
		\Vertex[L=\hbox{\tiny$\{b,b+1\}$},Lpos=90,x=1.5cm,y=1.5cm]{v3}
		
		\Vertex[L=\hbox{\tiny$\{\gamma_n,\gamma_n+1\}$},Lpos=90,x=4.5cm,y=1.5cm]{v6}
		\Vertex[L=\hbox{\tiny$\{\gamma_n,\gamma_n+1\}$},Lpos=-90,x=4.5cm,y=0.0cm]{v7}
		\Vertex[L=\hbox{\tiny$\{\gamma_{n+1},\gamma_{n+1}+1\}$},Lpos=-90,x=6.5cm,y=0.0cm]{v8}
		\Vertex[L=\hbox{\tiny$\{\gamma_{n+1},\gamma_{n+1}+1\}$},Lpos=90,x=6.5cm,y=1.5cm]{v9}
		
		\tikzset{VertexStyle/.style = {shape = circle, minimum size = 15pt, draw, color=white}}
		\Vertex[L=\hbox{$\ldots$},Lpos=-90,x=3.0cm,y=0.0cm]{v5}
		\Vertex[L=\hbox{\tiny$\ldots$},Lpos=90,x=3.0cm,y=1.5cm]{v4}
		
		\Edge[](v0)(v1)
		\Edge[](v0)(v2)
		\Edge[](v0)(v7)
		\Edge[](v0)(v8)
		
		\Edge[](v3)(v1)
		\Edge[](v3)(v2)
		\Edge[](v3)(v7)
		\Edge[](v3)(v8)
		\Edge[](v9)(v1)
		\Edge[](v9)(v2)
		
		\Edge[](v6)(v7)
		\Edge[](v6)(v8)
		\Edge[](v7)(v9)
		\Edge[](v8)(v9)
		\Edge[](v6)(v1)
		\Edge[](v6)(v2)
		
		\node at (v0.center) {\scriptsize$v_1$};
		\node at (v3.center) {\scriptsize$v_2$};
		\node at (v1.center) {\scriptsize$w_1$};
		\node at (v2.center) {\scriptsize$w_2$};
		\node at (v4.center) {\scriptsize$\ldots$};
		\node at (v5.center) {\scriptsize$\ldots$};
		\node at (v6.center) {\scriptsize$v_n$};
		\node at (v7.center) {\scriptsize$w_n$};
		\node at (v8.center) {\tiny$w_{n+1}$};
		\node at (v9.center) {\tiny$v_{n+1}$};
		\end{tikzpicture}
		\caption{Example of a complete bipartite graph with the color lists of partition $V=\{v_1,\ldots ,v_n\}$ equal to the partition $W=\{w_1,\ldots ,w_n\}$ of type $(\gamma, \mu)$ of size two.}
		\label{Figura:GrafoBipartidoLema1}
	\end{figure}
	
	Consider the complete subgraph of $ K_{n,n} $ obtained by the vertices $ v_1, v_2, \ldots,v_n $ and $ w_1, w_2,\ldots, w_n $. By induction hypothesis there is a list coloration for this $ K_{n,n} $. Coloring is missing vertices $ v_{n+1} $ and $ w_{n+1} $. In the worst case $ \gamma_{n+1} = \gamma_n+1 $ since $ \gamma_n<\gamma_{n+1} $ and the lists are in sequence. In this case, if $ \gamma_{n+1}=\gamma_n+1 $, without loss of generality, assume that color chosen for the vertex $ v_n $ was $ \gamma_n $ and for the vertex $ w_n $ was $ \gamma_n+1 $. Then the vertex $ v_{n+1} $ will be colored with the color $ \gamma_{n+1}+1 $ and the vertex $ w_{n+1} $ with the color $ \gamma_{n+1} $.
\end{proof}

\begin{lemma}\label{lemmatwo}
	Consider the complete bipartite graph $K_{n,n}=G[V,W]$ for $n\geq 1$ with partitions  $V=\{v_1,\ldots ,v_n\}$,$W=\{w_1,\ldots ,w_n\}$ and a distribution lists $\left( \gamma,\mu \right)$ of size $ 2 $, $L(v)$ for all $v\in K_{n,n}$ such that $L(v_i)\neq L(w_j)$ for all $i,j=1,\ldots,n$. Then with these lists, $K_{n,n}$ is $2$-$\left( \gamma,\mu \right)$-list-colorable.
\end{lemma}

\begin{proof}
Consider in partition $ V $ only the different lists $S_1,S_2,\ldots ,S_r$ and in partition $ W $ only the different lists $S_{r+1},S_{r+2},\ldots ,S_{r+t}$. Thus, we have to $ S_i\neq S_j $, for $ i\neq j $, and since they are lists $\left( \gamma,\mu \right)$ of size $ 2 $, the intersection of any two of these lists have at most one element.

\begin{figure}[ht]
	\centering
	\begin{tikzpicture}[scale=1]
	\SetGraphUnit{2}
	\GraphInit[vstyle=Normal]
	\SetUpVertex[LabelOut]
	\tikzset{VertexStyle/.style = {shape = circle, minimum size = 20pt, draw, color=black}}
	
	\Vertex[L=\hbox{\tiny$\{a,a+1\}$},Lpos=90,x=0.0cm,y=1.5cm]{v0}
	\Vertex[L=\hbox{\tiny$\{c,c+1\}$},Lpos=-90,x=0.0cm,y=0.0cm]{v1}
	\Vertex[L=\hbox{\tiny$\{d,d+1\}$},Lpos=-90,x=1.5cm,y=0.0cm]{v2}
	\Vertex[L=\hbox{\tiny$\{b,b+1\}$},Lpos=90,x=1.5cm,y=1.5cm]{v3}
	
	\Vertex[L=\hbox{\tiny$\{\gamma_n,\gamma_n+1\}$},Lpos=90,x=4.5cm,y=1.5cm]{v6}
	\Vertex[L=\hbox{\tiny$\{\gamma_n,\gamma_n+1\}$},Lpos=-90,x=4.5cm,y=0.0cm]{v7}
	\Vertex[L=\hbox{\tiny$\{\gamma_{n+1},\gamma_{n+1}+1\}$},Lpos=-90,x=6.5cm,y=0.0cm]{v8}
	\Vertex[L=\hbox{\tiny$\{\gamma_{n+1},\gamma_{n+1}+1\}$},Lpos=90,x=6.5cm,y=1.5cm]{v9}
	
	\tikzset{VertexStyle/.style = {shape = circle, minimum size = 15pt, draw, color=white}}
	\Vertex[L=\hbox{$\ldots$},Lpos=-90,x=3.0cm,y=0.0cm]{v5}
	\Vertex[L=\hbox{\tiny$\ldots$},Lpos=90,x=3.0cm,y=1.5cm]{v4}
	
	\Edge[](v0)(v1)
	\Edge[](v0)(v2)
	\Edge[](v0)(v7)
	\Edge[](v0)(v8)
	
	\Edge[](v3)(v1)
	\Edge[](v3)(v2)
	\Edge[](v3)(v7)
	\Edge[](v3)(v8)
	\Edge[](v9)(v1)
	\Edge[](v9)(v2)
	
	\Edge[](v6)(v7)
	\Edge[](v6)(v8)
	\Edge[](v7)(v9)
	\Edge[](v8)(v9)
	\Edge[](v6)(v1)
	\Edge[](v6)(v2)
	
	\node at (v0.center) {\scriptsize$v_1$};
	\node at (v3.center) {\scriptsize$v_2$};
	\node at (v1.center) {\scriptsize$w_1$};
	\node at (v2.center) {\scriptsize$w_2$};
	\node at (v4.center) {\scriptsize$\ldots$};
	\node at (v5.center) {\scriptsize$\ldots$};
	\node at (v6.center) {\scriptsize$v_n$};
	\node at (v7.center) {\scriptsize$w_n$};
	\node at (v8.center) {\tiny$w_{n+1}$};
	\node at (v9.center) {\tiny$v_{n+1}$};
	\end{tikzpicture}
	\caption{Example of a complete bipartite graph with all lists of different colors and of type $(\gamma, \mu)$ of size two.}
	\label{Figura:GrafoBipartidoLema2}
\end{figure}
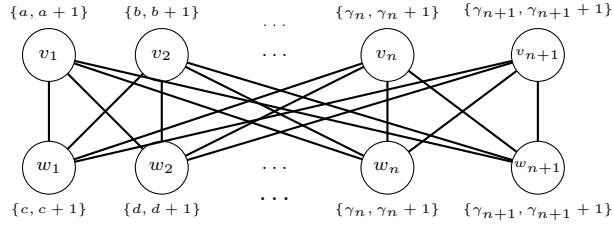

Thus, these $ S_i $ sets fulfill Hall's theorem and thus we obtain elements different $ c_i \in S_i $, with $ c_i\neq c_j $. Then we color a vertex $ v \in K{n,n} $ with the color $ c_i $ if $ L(v)=Si $, and thus we have a color list of $K{n,n}$.
\end{proof}

\begin{teo}\label{teo:bipartite}
	The $ K_{n,n} $ is $2$-$(\gamma$,$\mu)$-choosable.
\end{teo}

\begin{proof}
	Let $G=K_{n,n}$ a complete bipartite graph with partitions $V=\{v_1,v_2,\ldots,\\v_n\}$ e $W=\{w_1,w_2,\ldots ,w_n\}$, for $n\geq 1$. Consider a distribution of lists of colors of type $\left( \gamma,\mu \right)$ if size $ 2 $, $L(v)$ for every vertex $v$ de $G$. Let us choose the color lists that appear while the vertices of the two partitions $ V $ e $ W $. Clearly some of these lists may appear in more than one vertex in each partition, but will take just a list representative in her case appear more than once in the same partition. Thus, consider the sets $S_1,S_2,\ldots ,S_r$ $(r\leq n)$, which are the lists that occur while vertices of $ V $ and $ W $, and how we choose only one representative in the event of repetition in the partition, we have $S_i\neq S_j$, for $i\neq j$. If $r=0$ ,the graph $ G $ satisfies the hypotheses of Lemma 1 and so we have a list coloring of $ G $. If $r\geq 1$ can form a complete subgraph $K_{r,r}$ with these color lists  $S_1,S_2,\ldots ,S_r$ and a vertex for each representative. This $K_{r,r}$ graph $ G $ satisfies the hypotheses of Lemma \ref{lemmatwo} and so we have a list coloring of $K_{r,r}$.
\end{proof}

A stronger proof in $k$-$(\gamma$,$\mu)$-choosable will be presented, which involves the more general coloring of graphs.

\begin{teo}
	If $G$ have a $k$-coloring, so $G$ is $k$-$(\gamma,\mu)$-choosable.
\end{teo}\label{teo:general}

\begin{proof}
	When we do the division of any integer $m$ by $k$, we get the rest of the division $r$, such that $r\in \{0,1,2,\ldots,k-1\}$. Besides this division partitions the integers into $ k $ distinct sets.	
	Consider $A_0=\{n\in \mathbb{Z}|$ the remainder of dividing $ n $ by $ k $ is $0$ $\}$, $A_1=\{n\in \mathbb{Z}|$the rest of the division of $ n $ by $ k $ is $1$ $\}$, $A_2=\{n\in \mathbb{Z}|$the rest of the division of $ n $ by $ k $ is $2$ $\}\ldots$, $A_{k-1}=\{n\in \mathbb{Z}|$the rest of the division of $ n $ by $ k $ is $k-1$ $\}$. These sets two by two dijuntos, $A_i\cap A_j=\emptyset$, if $i\neq j$. Another important fact is that any given sequence of  $k$ consecutive integers $x, x +1, x +2, \ldots, x + k-1 $, there exists exactly one of these elements contained in each of the sets $A_i$. That is, one of these elements leaves a remainder $ i $ in the division by $ k $ for $ i = 0,1,2, \ldots k-1 $.
	
\vspace{-.3cm}	
	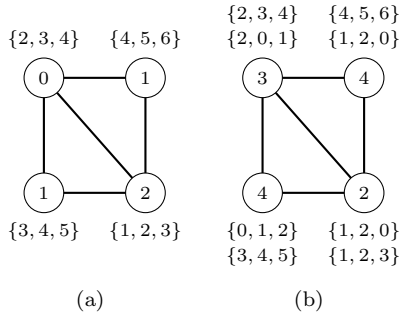
\begin{figure}[h]
		\centering
		\subfigure[]
		{
			\begin{tikzpicture}[scale=.9]
			\SetGraphUnit{2}
			\GraphInit[vstyle=Normal]
			\SetUpVertex[LabelOut]
			\tikzset{VertexStyle/.style = {shape = circle, minimum size = 15pt, draw, color=black}}
			
			\Vertex[L=\hbox{\scriptsize$\{2,3,4\}$},Lpos=90,x=0.0cm,y=1.7cm]{v0} 
			\Vertex[L=\hbox{\scriptsize$\{3,4,5\}$},Lpos=-90,x=0.0cm,y=0.0cm]{v1}
			
			\Vertex[L=\hbox{\scriptsize$\{1,2,3\}$},Lpos=-90,x=1.5cm,y=0.0cm]{v2}
			\Vertex[L=\hbox{\scriptsize$\{4,5,6\}$},Lpos=90,x=1.5cm,y=1.7cm]{v3} 

			\Edge[](v0)(v1) 	
			\Edge[](v0)(v2)		
			\Edge[](v0)(v3)		
			\Edge[](v2)(v3)		
			\Edge[](v1)(v2)

			\node at (v0.center) {\scriptsize$0$};
			\node at (v3.center) {\scriptsize$1$};
			\node at (v1.center) {\scriptsize$1$};
			\node at (v2.center) {\scriptsize$2$};
			
			\node (v0) [above of=v0, left=0cm, above=-0.4cm] {\scriptsize$\textcolor{white}{\{2}\textcolor{white}{3}\textcolor{white}{4\}}$};
			\node (v3) [above of=v3, left=0cm, above=-0.4cm] {\scriptsize$\textcolor{white}{\{\}4}\textcolor{white}{5}\textcolor{white}{6\}}$};
			\node (v1) [below of=v1, left=0cm, below=-0.4cm] {\scriptsize$\textcolor{white}{\{\}3}\textcolor{white}{4}\textcolor{white}{5\}}$};
			\node (v2) [below of=v2, left=0cm, below=-0.4cm] {\scriptsize$\textcolor{white}{\{\}1}\textcolor{white}{2}\textcolor{white}{3\}}$};
			\end{tikzpicture}
			\label{figure1}
		}
		\subfigure[]
		{
			\begin{tikzpicture}[scale=.9]
			\SetGraphUnit{2}
			\GraphInit[vstyle=Normal]
			\SetUpVertex[LabelOut]
			\tikzset{VertexStyle/.style = {shape = circle, minimum size = 15pt, draw, color=black}}
			
			\Vertex[L=\hbox{\scriptsize$\{2,0,1\}$},Lpos=90,x=0.0cm,y=1.7cm]{v0} 
			\Vertex[L=\hbox{\scriptsize$\{0,1,2\}$},Lpos=-90,x=0.0cm,y=0.0cm]{v1}
			
			\Vertex[L=\hbox{\scriptsize$\{1,2,0\}$},Lpos=-90,x=1.5cm,y=0.0cm]{v2}
			\Vertex[L=\hbox{\scriptsize$\{1,2,0\}$},Lpos=90,x=1.5cm,y=1.7cm]{v3} 

			\Edge[](v0)(v1) 	
			\Edge[](v0)(v2)		
			\Edge[](v0)(v3)		
			\Edge[](v2)(v3)		
			\Edge[](v1)(v2)		
			
			\node at (v0.center) {\scriptsize$3$};
			\node at (v3.center) {\scriptsize$4$};
			\node at (v1.center) {\scriptsize$4$};
			\node at (v2.center) {\scriptsize$2$};
			
			\node (v0) [above of=v0, left=0cm, above=-0.4cm] {\scriptsize$\{\textcolor{black}{2},\textcolor{black}{3},\textcolor{black}{4}\}$};
			\node (v3) [above of=v3, left=0cm, above=-0.4cm] {\scriptsize$\{\textcolor{black}{4},\textcolor{black}{5},\textcolor{black}{6}\}$};
			\node (v1) [below of=v1, left=0cm, below=-0.4cm] {\scriptsize$\{\textcolor{black}{3},\textcolor{black}{4},\textcolor{black}{5}\}$};
			\node (v2) [below of=v2, left=0cm, below=-0.4cm] {\scriptsize$\{\textcolor{black}{1},\textcolor{black}{2},\textcolor{black}{3}\}$};
			
			\end{tikzpicture}
			\label{figure2}
		}
		\caption{The Figure \ref{figure1} presents a $ 3 $-colored graph with colors $ 0,1 $ and $ 2 $ and the Figure \ref{figure2} presents a choice of colors for its vertices in an appropriate way, according to the Theorem \ref{teo:general}, therefore, the vertex with color $ 0 $ receives the color $ 3 $ of the list $ \left(\gamma, \mu \right)$, the vertices with the color $ 1 $ receive the color $ 4 $ and the vertex with the color $ 2 $ also receives the color $ 2 $ of the list $ \left(\gamma, \mu \right)$.}
		\label{fig:exemplo_4cores}
	\end{figure}
	
	Now consider a given distribution of lists of size $ k $ of $\left(\gamma,\mu \right)$'s type to the vertices of $ V $. Thus, for each $ v_i \in V $, we have $ L (v_i) = \{\gamma_i, \gamma_i+1, \ldots, \gamma_i + k-1 \} $. Let $ v_i $ and $v_j $  adjacent vertices in $ G $, and $ c_i = c(v_i) \neq c(v_j) = c_j$, and therefore the sets $ A_{c_i} $ and $ A_{c_j}$  are disjoint. To do one $ \left (\gamma, \mu \right) $-list coloring of $ G $ in these adjacent vertices $ v_i $ and $v_j $  choose as color for the vertex $ v_i $ element of $ L(v_i) $ that belongs to the set $ A_{c_i} $ and vertex $v_j $ choose the element of $ L(v_j) $ that belong to the set $ A_{c_j} $ and this ensures that adjacent vertices are colored with distinct colors. Therefore $ G $ is $k$-$\left (\gamma, \mu \right) $-choosable.
\end{proof}

 The choosability in graphs is a $\Pi_2^P$-complete problem \cite{article-erdos}, however, we saw in the previous theorem that $k$-$(\gamma,\mu)$-choosability is an $NP$-problem due to its relation with the $k$-coloring.

\section{Conclusions}

A $ (\gamma,\mu) $-coloring is a graph coloring problem that has been studied frequently in recent years. It has a real-world application in an interesting and well-studied way: the channel assignment problems \cite{dias2014modelos}. The correlation between $(\gamma$,$\mu)$-coloring and choosability in graphs was very interesting because in this case, the choosability stopped being $\Pi_2^P$-complete to be $NP$-complete. The use of methods of proof used for choosability in some classes of graphs is interesting in the sense of presenting the diversity of forms of solutions to this problem. The goal is to continue working with $(\gamma, \mu)$-coloring on other forms of complexity. Table \ref{table:results} presents a comparison with some results in classical coloring of graphs, choosability in graphs and the results obtained in this work with $k$-$\left(\gamma,\mu \right)$-choosable (in asterisk), as a consequence of Theorem \ref{teo:general}.

\begin{table}[ht]
	\centering
	\label{table:results}
	\footnotesize{
		\begin{tabular}{|p{1.9cm}|p{1.5cm}|p{2.2cm}|P{2.6cm}|}
			\hline
			Graph Classes & \centering $\chi(G)$ & \centering$\chi_\ell(G)$                 & $\chi_{\ell(\gamma,\mu)}$           \\ \hline
			Trees         &  $ 2 $-colorable     & $ 2 $-choosable \cite{book-chartrand}    & $ 2 $-$(\gamma,\mu)$-choosable$ ^* $\\ 
			Even Cycles   &  $ 2 $-colorable     & $ 2 $-choosable \cite{article-erdos}     & $ 2 $-$(\gamma,\mu)$-choosable$ ^* $\\
			Bipartite     &  $ 2 $-colorable     & $ 3 $-choosable \cite{article-erdos}     & $ 2 $-$(\gamma,\mu)$-choosable$ ^* $\\  
			Odd Cycle     &  $ 3 $-colorable     & $ 3 $-choosable \cite{article-erdos}     & $ 3 $-$(\gamma,\mu)$-choosable$ ^* $\\ 
			Planar        &  $ 4 $-colorable     & $ 5 $-choosable \cite{thomassen1994every} & $ 4 $-$(\gamma,\mu)$-choosable$ ^* $\\ \hline	
	\end{tabular}}
	\caption{Some results in graph coloring, choosability and $ k $-$(\gamma,\mu)$-choosability obtained in this work.}
\end{table}

\vspace{-.5cm}
\section{Acknowledgment}

Thanks to CAPES (Coordena\c{c}\~ao de Aperfei\c{c}oamento de Pessoal de N\'ivel Superior), CNPq (Conselho Nacional de Desenvolvimento Cient\'ifico e Tecnol\'ogico) and FAPEAM (Funda\c{c}\~ao de Amparo a Pesquisa do Estado do Amazonas) for the funding of this work.






\renewcommand{\refname}{6. Reference}

\end{document}